\documentclass{amsart}[12pt]
\usepackage{amsmath,amscd}
\usepackage[]{listings}
\usepackage[]{graphicx}
\usepackage[]{xcolor}
\usepackage[]{url}

\newtheorem{lemma}{Lemma}

\newtheorem{open}{Open problem}

\definecolor{darkgreen}{rgb}{0.0, 0.27, 0.13}
\definecolor{forestgreen}{rgb}{0.0, 0.27, 0.13}
\definecolor{shellcolor}{RGB}{205,192,176}
\definecolor{bashcolor}{RGB}{205,192,176}
\definecolor{makecolor}{RGB}{192,192,192}
\definecolor{flexcolor}{RGB}{180,205,205}
\definecolor{yacccolor}{RGB}{205,181,205}
\definecolor{termcolor}{gray}{0.80}
\definecolor{tracecolor}{gray}{0.90}
\definecolor{gcccolor}{RGB}{224,255,255}
\definecolor{algocolor}{RGB}{144,238,144}
\definecolor{pythoncolor}{RGB}{144,238,144}
\definecolor{asmcolor}{RGB}{230,230,250}

\lstdefinestyle{gcc}{%
frame=single,
xleftmargin=2em,
language=c,
numbers=left,
numberstyle=\footnotesize,
columns=[l]flexible,
tabsize=4,
showtabs=false,
backgroundcolor=\color{gcccolor},
stringstyle=\color{red}\itshape,
showstringspaces=false,
keywordstyle=\color{blue}\bfseries,
commentstyle=\color{forestgreen}\itshape,
emphstyle=\color{red}\bfseries,
emph={[2]os,sys},
emphstyle={[2]\color{pink}\bfseries},
literate={{:cup:}{$\ \cup\ $}1{:lamb:}{$\lambda$}1{:su:}{$s_u$}1{:rond:}{$\circ$}1{:in:}{$\in$}1{:=}{$\gets$}1{:notin:}{$\not\in$}1{:emptyset:}{$\emptyset$}1}
}

\def\rmc(#1,#2){RM(#1,#2)}
\def\form(#1,#2){H(#1,#2)}
\def\val(#1,#2){V(#1,#2)}
\def\boole(#1){ B(#1) }

\def\fd{{\mathbb F}_2}
\def\fdm{{\mathbb F}^m_2}

\def\aglm{{\textsc{agl}}(m,2)}

\def\agl#1#2{\textsc{agl}(#1,#2)}

\def\stab(#1){\textsc{stab}(#1)}
\def\stablevel(#1,#2){\textsc{stab}_{#1}(#2)}

\def\cls{{\rm n}}
\def\rmq(#1,#2,#3){\rmc(#1,#2)/\rmc(#3,#2)}

\def\level(#1){\underset{#1}{=}}
\newcommand{\binomial}[2]{\genfrac{(}{)}{0pt}{}{#1}{#2}}
\newcommand{\card}[1]{\vert{#1}\vert}
\def\val(#1){{\rm val}(#1)}

\def\agl#1{{\mathfrak #1}}
\def\orb#1{{\mathcal O}_#1}
\def\anf(#1){{\rm anf}(#1)}
\def\fix(#1,#2,#3,#4){{{\rm fix}^{#1,#2}_{#3}}{(#4)}}
\def\classe(#1,#2,#3){{{\mathcal T}(#1,#2,#3)}}
\def\rset#1{{\mathcal #1}}
\def\stab(#1){\textsc{stab}(#1)}
\def\stablevel(#1,#2){\textsc{stab}^{#1}(#2)}
\def\stableveldim(#1,#2,#3){\textsc{stab}^{#1}_{#2}(#3)}
\def\stableveldeg(#1,#2,#3,#4){\textsc{stab}^{#1,#2}_{#3}(#4)}
\def\level(#1){\underset{#1}{\sim}}
\def\bound(#1,#2){\underset{#1}{\overset{#2}{\sim}}}
\def\modulo(#1,#2){\mod\rmc(#1,#2)}

\def\pow#1.#2{\tiny$10^{#1.#2}$}

\begin{document}
\title{Classification of some cosets of the Reed-Muller code} 

\author[1]{Valérie Gillot}
\email{valerie.gillot@univ-tln.fr}
\author[]{Philippe Langevin}
\email{philippe.langevin@univ-tln.fr}
\address{Imath, universit\'e de Toulon}
\thanks{This work is partially supported by the French Agence Nationale de la Recherche through the SWAP project under Contract ANR-21-CE39-0012}

	\begin{abstract} 
	This paper presents a descending method to classify Boolean functions in 7 variables under the action of the affine general linear group. The classification determines the number of classes, a set of orbits representatives and a generator set of the stabilizer of each representative. The method consists in the iteration of the classification process of $\rmq( k, m,r-1)$ from that of $\rmq( k, m,r)$. We namely obtain the classifications of $\rmq( 4, 7,2)$ and of $\rmq(7,7,3)$, from which we deduce some consequences on the covering radius of $\rmc(3,7)$ and the classification of near bent functions. 
	
\end{abstract}

\maketitle

\section{Introduction}

Let $\fd$ be the finite field of order $2$. Let $m$ be a positive integer. 
A mapping from $\fdm$ into $\fd$ is called a Boolean function. Every Boolean 
function  has a unique algebraic reduced representation :
$$
f(x_1, x_2, \ldots, x_m ) = f(x) = \sum_{S\subseteq \{1,2,\ldots, m\}} a_S X_S,
\quad a_S\in\fd, \ X_S( x ) = \prod_{s\in S} x_s.
$$
The degree of $f$ is the maximal cardinality of $S$ with  $a_S=1$ in
the algebraic form.  The valuation of $f\not=0$, denoted by $\val(f)$,
is the minimal cardinality of $S$ for
which $a_S=1$. Conventionally,  $\val(0)$ is $\infty$. 
We denote by $\boole(s, t, m)$ the space of Boolean
functions of valuation greater than or equal to $s$ and of degree less 
than or equal to $t$. Note that  $\boole(s,t,m)=\{0\}$ whenever $s>t$.
The space $\boole(0, t, m)$ identifies with the Reed-Muller code
$\rmc(t,m)$ 
and $\boole(s, t, m)$ is the representation 
of the quotient space $\rmq( t, m,s-1)$. The affine general 
linear group of $\fdm$, denoted by $\aglm$, acts naturally 
over all these spaces. The number of classes
of $\boole( s, t, m)$, denoted by  $\cls(s, t , m)$, 
satisfies a nice duality relation~:

\begin{equation}
	\label{CLASSREL}
	\cls(s, t, m ) = \cls( m-t, m-s, m).
\end{equation}

X.-D.~Hou gives a proof of the above relation in
\cite{AGLHOU}. In the proof of Lemma \ref{DUALITY}, 
we propose an alternative demonstration.

For the dimensions that we want to consider, all
class numbers are very easy to determine 
using  Burnside's Lemma and the theory of conjugacy 
classes of $\aglm$, see e.g. \cite{BOOK}.

In general, such a class number is huge, but, 
when it is reasonably small, one may consider 
to determine an orbit representative 
set that is a list of $\cls( s, t, m )$ Boolean functions,
of degree less than or equal to $t$, and  pairwise non affine 
equivalent modulo $\rmc(s-1,m)$. As an example,  the class 
number $\cls( 2, 6, 6)$ is 
$150357$ and J.~Maiorana in \cite{MAIORANA} describes a recursive algorithm 
to find the 150357 equivalence classes.

More generally, the classification data
of the space $\boole(s, t, m )$ plays an important role
both in coding theory and cryptography. The covering
radii of Reed-Muller codes are not generally known and the
classification of $\boole(s, t, m )$  can be used to bound the covering 
radius of $\rmc(s-1 , m)$ in $\rmc(t,m)$ as in the paper \cite{WANG}. These  
classifications are also used to study the cryptographic
parameters of Boolean functions.

This paper presents a procedure to provide 
classifications of Boolean functions spaces
for $m=7$. Precisely, we compute orbit representative sets
of $\boole(s,t,7)$, for all parameters $s\leq t\leq 7$
such that  $\cls(s , t , 7)$ is less than $10^6$. 

Our approach gives \emph{complete} 
classifications : not only sets of orbit representatives,
but also for each representative, a generator set of stabilizer group.
The most interesting cases are the classifications of $\boole(3,4,7)$ and of $\boole(4,7,7)$. From the first one, we determine the classification of near bent functions. From the second, we refind, with an alternative method, the covering radius of $\rmc(3,7)$ obtained in \cite{COVERING}.

All computed data are available on the project page \cite{PROJECT}.

\section{Boolean functions}

A Boolean function $f$ is a member
of $\boole(s,t,m)$ if and only if $s\leq \val(f)$ and  $\deg(f)\leq t$. 
Denoting $\bar S$ the complement  set of $S\subseteq\{1,2,\ldots,m\}$, 
the complementary transform $\sum_S X_S \mapsto \sum_S X_{\bar S}$ maps 
$\boole(s,t,m)$ onto $\boole(m-t,m-s,m)$, in particular, 
these spaces have the same dimension. A Reed-Muller code of order $k$ in $m$ variables is the space of Boolean functions of degree less or equal to $k$~:
$$\rmc(k,m)=\{ f\in \boole(m) \mid \deg(f) \leq k\}.$$
Note that Reed-Muller spaces are nested~:
\begin{equation*}
	\underbrace{\rmc(-1,m)}_{(0)}\subset \rmc(0,m) \subset \rmc(1,m) \subset \cdots 
	\subset \rmc(m-1,m) \subset \underbrace{\rmc(m,m)}_{\boole(m)}.
\end{equation*}

The quotient space $\rmq( k, m, k-1)$ is the space of homogeneous forms of degree
$k$, identified with the space $\boole(k,k,m)$, its dimension is the value of the binomial 
coefficient $\binomial mk$. The dimension 
of $\boole(s,t,m)$ is equal to the sum of binomial coefficients 
$\sum_{k=s}^t \binomial mk$.  It is easy to see that the weight
of a Boolean function is even if and only if its degree is not maximal,
consequently the orthogonal of $\rmc(k,m)$ is  $\rmc(m-k-1,m)$,
with respect to  the scalar product $\langle f, g \rangle = \sum_{x\in\fdm} f(x) g(x)$.

\begin{lemma}[duality]
\label{DUALITY}
	For all $s,t$  such that $s\leq t\leq m$, $\boole(m-t,m-s,m)$ is
	a representation of $\boole(s,t,m)^\ast$,  the dual space of $\boole(s,t,m)$.
	It means that for any form $\phi\in\boole(s,t,m)^\ast$ there exists
	one and only one $g\in\boole(m-t,m-s,m)$ such that 
	$\phi(f) =\langle f, g \rangle$, for all $f\in\boole(s,t,m)$.
\end{lemma}
\begin{proof}
	Note that the dimension of $\boole(m-t,m-s,m)$ is precisely
	the dimension of $\boole(s,t,m)$. 
	If $0 \not= g\in\boole(m-t,m-s,m)$ then
	$g\not\in\boole(s, t, m)^\perp$. Indeed, consider
	a monomial term $X_S$ of maximal degree in the algebraic
	representation 
	of $g$~: 
	$$m-t\leq \deg(g)=\deg(X_S) \leq m-s \quad\text{and} \quad s\leq \deg(X_{\bar S}) \leq t.$$
	The product $X_{\bar S}g=X_{\bar S}X_{S}+\cdots$ has degree $m$
	whence $X_{\bar S}$ is member
	of $\boole(s, t, m)$ which is not orthogonal to 
	$g$. In other words, the space $\boole(m-t,m-s,m)$ 
	is a representation of $\boole(s,t,m)^\ast$.
\end{proof}

\section{Action of the affine general linear group}

First, let us recall some definitions. Let $(G,*)$ be a finite group and let $U$ be a finite set, a right group action of $G$ on $U$ is a mapping from $U \times G $ to $U$ denoted by $(u,g)  \longmapsto u\circ g$ satisfying $u\circ e =u$  and $(u \circ g) \circ h = u \circ (g*h)$, for $u\in U$, $g,h\in G$ and $e$ the identity of $G$.
The orbit of an element $u$  is the set of elements in $U$ to which $u$ can be moved by the elements of $G$, denoted by $\orb u = \{ u\circ g \mid g \in G \}$.
The stabilizer subgroup of $G$ with respect to $u\in U$ is the set of elements in $G$ that fixes $u$, denoted by $\stab(u)=\{ g \in G \mid u \circ g = u \}$.

The affine general linear group acts naturally on the right
over Boolean functions. The action of $\agl s\in\aglm$
on a Boolean function $f$ is $f\circ \agl s$, the composition of applications. The order of $\aglm$  is $2^m\prod_{i=0}^{m-1} (2^m - 2^i )
\approx 0.29\ 2^{m^2+m}$. Note that
the number of orbits of this group action  has doubly exponential growth 
with the parameter $m$. For $m=7$, it is already numerically 
impossible to list the $\approx 2^{74}$ classes of Boolean functions !  

The Reed-Muller spaces are invariant under the action
of $\aglm $.
Considering the action modulo $\rmc(r,m)$, the space of functions of degree less or equal than $r$, we introduce objects at level $r$.
Two  Boolean functions $f$ and $g$ in $m$ 
variables are equivalent at level $r$,
if there exists $\agl s\in\aglm$ such that   $ f\circ \agl s\equiv g \mod \rmc(r,m)$.
We introduce two notations  $f\level(r)g$ for the equivalence
at level $r$, and $\stableveldim( r, m, f)$, for the
stabilizer of $f$ at level $r$ :
\begin{equation}
	\label{EQUIV}
	f\level(r)g \Longleftrightarrow \exists \agl s \in \aglm,\ f\circ \agl s\equiv g \mod \rmc(r,m).
\end{equation}

\begin{equation}
	\label{STABILIZER}
	\stableveldim( r, m, f) =\{ \agl s\in\aglm \mid f\circ \agl s  \equiv f \mod \rmc(r,m) \}.
\end{equation}
In this paper, we consider the action of $\aglm$ over $\boole(s,t,m)$ as the composition of applications modulo $\rmc(s-1,m)$. Precisely, two elements $f,g\in \boole(s,t,m)$ are in the same orbit, under this action, if and only if they are equivalent at level $s-1$, that is $f\level(s-1) g$. In this context, the stabilizer of $f$ is nothing but $\stableveldim(s-1,m,f)$ the stabilizer at level $s-1$.

 Thus, the  affine general linear group acts over the $\boole(s, t, m)$,
the corresponding class number $\cls(s,t,m)$ is given by 
Burnside's formula~:

\begin{equation}
\label{BURNSIDE}
	\card{\aglm} \times \cls(s,t,m) = \sum_{\agl s\in \aglm} \sharp\fix(s,t,m,\agl s)
	= \sum_{\agl s\in \Gamma} R(\agl s)\,\sharp\fix(s,t,m,\agl s).
\end{equation}
where $\fix(s,t,m, \agl s)$ is the set $\{ f\in \boole(s,t,m) \mid f\circ \agl s \equiv f \mod \rmc(s-1,m)\}$, i.e. the kernel of the endomorphism 
of $\boole(s,t,m)$ defined by $f\mapsto f\circ\agl s$. In practice, we reduce the sum to the
$\Gamma$, a set of representatives of conjugacy classes of
$\aglm$, and $R(\agl s)$ the size of the conjugacy class of $\agl s$,
see book \cite{BOOK} for the finite fields combinatoric details.
\begin{lemma}[formula]
	For all $s,t$  such that $s\leq t\leq m$, 
        $$\cls(s,t,m)  = \cls(m-t,m-s,m)$$
\end{lemma}
\begin{proof}
The number of orbits of a finite space $E$ under  the action of a subgroup $G$ of the general linear group $\text{aut}(E)$  is the same that the number of orbits of the dual group $G^\ast$. The Lemma statement is a particular case of this result.
For $\agl s\in\aglm$, the adjoint of 
	the automorphism $f\mapsto f\circ\agl s$  corresponds to
	the inverse of $\agl s$, because 
	$$
	\langle f\circ\agl s, g \rangle = \sum_{x\in\fdm} f\circ\agl s(x) g(x) = \sum_{x\in\fdm} f(x) g\circ{\agl s}^{-1} (x) =
	\langle f, g\circ{\agl s}^{-1}\rangle.
	$$
The result follows using Burnside's formula by observing 
\begin{align*}
	\sum_{\substack{f\in \boole(s,t,m)\\g\in \boole(s,t,m)^\ast}} (-1)^{   \langle f\circ\agl s + f, g \rangle}& =\sum_{\substack{f\in \boole(s,t,m)\\g\in \boole(s,t,m)^\ast}} (-1)^{\langle f, g\circ {\agl s}^{-1} +g \rangle}
\\
\sharp \boole(s,t,m)^\ast \times \sharp\fix(s,t,m,\agl s)&= \sharp \boole(s,t,m) \times \sharp\fix(m-t,m-s,m,{\agl s}^{-1})\\
\end{align*}

\end{proof}

In this paper, by a classification at level $r$ of degree $k$
in $m$ variables, we mean a classification of $\boole(r+1,k,m)$,  that is a set of orbit representatives 
at level $r$ under the right action of $\aglm$, and for each 
orbit representative $f$, a  generator set of $\stableveldim( r, m, f)$, the stabilizer of $f$ at level $r$.   It is important
to note that at level $r$, we calculate modulo $\rmc(r,m)$,
and we consider polynomials whose valuations are strictly 
greater than $r$.

Recall that $\aglm$ can be  generated by three following transformations of $v=(v_m,v_{m-1},\ldots, v_1)$~: 
the shift operator $S\colon v \mapsto (v_{m-1}, \ldots, v_1, v_m)$,
the transvection $T\colon v \mapsto (v_m,\ldots,v_2,v_1+v_2)$
and the translation $U\colon v \mapsto v+(0,\ldots,0,1)$.

In next section, we detail the procedure that 
we used to build  a classification at level $r-1$
from a classification at level $r$. Starting at level $k$, there is only one orbit $\{0\}=\boole(k+1,k,m)$
stabilized by full group $\aglm =\langle S,T,U\rangle$. One can start from this
classification at level $k$ to determine the classifications 
at level $k-1$, level $k-2$, etc. The process can be stopped at any level or be continued  until level $-1$ to reach the classfication of $\boole(0,k,m)=\rmc(k,m)$. In this way,  we classify $\boole(s,t,m)$ in $t-s+1$ iterations starting from the classification of $\boole(t+1,t,m)=\{0\}$.

\section{Descending procedure}

In order to deduce a classification at level $r-1$ 
from a classification a level $r$, we have to consider 
some  ``boundary actions'' on $\boole(r,r,m)$ the space of 
homogeneous forms of degree $r$. 

An element $\agl s$ of the stabilizer of $f$ at level $r$ induces an action 
on homogeneous forms of degree $r$
defined  for $u\in \boole(r,r,m)$ by  $$ u \mapsto u\circ\agl s + f\circ{\agl s} +f \modulo(r-1,m) $$

\begin{lemma}[boundary]
	\label{BOUNDARY}
	Let $\rset R$ be a set of orbit representatives of degree
	$k$ at level $r$. For each $f\in\rset R$, $\rset U(f)$ denotes
	a set of orbit representatives of $\boole(r,r,m)$ under  
	the boundary action of $\stableveldim(r,m, f)$. We obtain that  $\{ f+u\mid f\in\rset R, u\in\rset U(f) \}$ is a set of orbit 
	representatives with same degree at level $r-1$.
\end{lemma}
\begin{proof}
	We start by showing the elements of this set are not equivalent at level $r-1$.
	Indeed, let  $f'$ and $f$ be in $\rset R$, and two forms
	$u'\in\rset U(f')$ and $u \in\rset U(f)$ such that $f+u\level(r-1) f'+u'$.  
	There exists $\agl  s\in\aglm$ such that $f'+u'\equiv(f+u)\circ\agl  s \modulo(r-1,m)$. 
	Reducing more, we obtain 
	$f' \equiv f\circ \agl s \modulo(r,m)$; so that $f'$ and $f$ are
	equivalent at level $r$, thus $f'=f$. The boundary action of $\agl s\in\stableveldim(r,m,f)$
	sends $u$ to $u'$ and finally $u'=u$.
	Now, we prove that the  set represents all polynomials at level $r-1$. Indeed,
	for $g\in\boole(r-1, k, m)$, there exists a 
	pair $(\agl t, f)\in\aglm\times\rset R$ 
	such that $g\circ \agl t \equiv f\modulo(r,m)$, whence 
	$g\circ \agl t \equiv f+v\modulo(r-1,m)$, 
	where $v$ is a form of degree $r$.
	Moreover, there is a boundary action $\agl s\in\stableveldim(r,m,f)$ that sends
	$v$ to some $u\in\rset U(f)$ 
	whence $g \circ \agl t\agl s \equiv (f+v)\circ\agl s \equiv f+u\modulo(r-1,m) $. 
\end{proof}

For a right action of a group $G$ on a set $U$ and $u\in U$, we denote by $\orb u$  the orbit of $u$, 
$S_u$ the stabilizer of $u$ and
$s_u$ the order of $S_u$.

\begin{lemma}[class formula] 
	\label{CLASS} 
	If $G$ is a finite group acting  on a finite set $U$ then
	the size of the orbit of an element $u\in U$ is equal to $\card G / s_u$.
\end{lemma}
\begin{proof}
	There is a bijection from $G/S_u$ onto $\orb u$ the orbit of $u$.
\end{proof}
\begin{lemma}[Schreier]
	\label{SCHREIER}
	Let $L$ be a set of generators
	of a finite group $G$ right acting on 
	a finite set $U$. Let $\orb u$  be the orbit of some
	element $u\in U$.  If $R\colon \orb u \rightarrow G$ 
	is a map such that $u\circ R(x)  = x$ for all $x\in\orb u$
	then $\{ R(x)\lambda R(x\circ\lambda)^{-1}  \mid \lambda\in L,  x\in \orb u \}$
	generates the stabilizer $S_u$  of $u$.
\end{lemma}
\begin{proof}
	See \cite{SERESS}.	
\end{proof}

Knowing the value $s_u$,  one can build a generator set 
of its stablizer $S_u$ applying  Schreier's Lemma. We implement 
this idea in the algorithm \texttt{generatorSet} where $\ast$ denotes
the law group and $\circ$ denotes the action of the group.

\begin{center}
	\begin{lstlisting}[float,style=gcc, xleftmargin={20mm},caption={Construction of a generator set of $S_u$.},linewidth={110mm},numbers=left, emph={pop,push,empty}]
Algorithm  generatorSet( u , L, :su:  )
{	// return a generator set of the stabilizer of u
	// under the action of the group generated by L
	// knowing its order :su:
	S := :emptyset:
	push( u )
	R [ u ] :=  id
	Y :=  { u } 
	while (  order( <S> ) < :su:  ) {
		pop( x )
		for :lamb: :in: L { 
			y :=  x :rond: :lamb:
			if y :notin: Y {
				push(y)
				R[ y ] := R[ x ] * :lamb:
				Y:= Y :cup: {y} 
			} else {
			  s :=  R[x] * :lamb: * inverse( R[ y ] )
			  if ( s not in <S> )
			     S := S :cup: { s }
			 }
		}
    }
	return S;
}
\end{lstlisting}
\end{center}

Now, we describe our descending procedure based on Lemma \ref{BOUNDARY}
and Lemma \ref{SCHREIER} to construct a set of  orbit representatives at level $r-1$ from level $r$. 
In view of dimension of forms  space  $\boole(r,r,m)$ and to save memory space, 
we proceed in two phases~:
\begin{enumerate}
	\item For each representative $f$ at level $r$, we use 
         a classical algorithm to enumerate an orbit representatives set 
	of $\boole(r,r ,m)$ under the action of $\stableveldim(r,m,f)$. For each
	representative $u$,  we obtain  the orbit $\orb u$, and by Lemma \ref{CLASS}, 
	the order $s_u$ of $\stableveldim(r-1, m, f+u)$ is equal to 
	$\sharp \stableveldim(r, m, f ) / \sharp \orb u$. 
\item  For each representative $f$ at level $r$, let $L$ be a generator set of 
	$\stableveldim(r,m,f)$. For each pair $(u, s_u)$, obtained in (1), 
	we apply \texttt{generatorSet($u$, $L$, $s_u$)}  to construct
	a set of generators of  $\stableveldim(r-1,m,f+u)$.
\end{enumerate}

\section{results and applications}

Our implementation in C language of the descending procedure,
without any parallelization,  builds the full classification 
of $\boole(2,6,6)$ in 15 secondes. It classifies 
$\boole(3,4,7)$ in three days by requiring about 50GB of memory. 

The values of $\cls(s,t,7)$ for $0 \leq s\leq t\leq 7$ are listed 
in Table \ref{CLASSNUMBER}. For all parameters $0\leq s\leq t\leq 7$
such that  $\cls(s , t , 7) < 10^6$, the descending procedure
classifies $B(s,t,7)$, it computes for each orbit, a representative 
and also a generator sets of the corresponding stabilizer. 
All the
numerical data are available in project page \cite{PROJECT}. In the next subsections,
we focus on applications of the classifications of $\boole(3,4,7)$  
and $\boole(4,7,7)$.

\begin{table}
	\caption{\label{CLASSNUMBER}Class numbers $\cls(s,t,7)$.}
\begin{tabular}{|c|ccccccc|}
	\hline
$s\backslash t$  &1 &2 &3 &4 &5 &6 &7\\
\hline
\hline
0   &3 &12 &3486  &\pow13.5 &\pow19.8 &\pow21.9 &\pow22.2\\
1   &2 &8 &1890   &\pow13.1 &\pow19.5 &\pow21.6 &\pow21.9\\
2   &      &4 &179 &\pow 11.0   &\pow17.3  &\pow19.5 &\pow19.8\\
3   &      &      &12  &68443 &\pow11.0   &\pow13.1  &\pow13.5\\
4   &      &      &        &12    &179    &1890  &3486\\
	5   &      &      &        &          &4       &8     &12\\
6   &      &      &        &          &           &2     &3\\
7   &      &      &        &          &           &          &2\\
\hline
\end{tabular}
\end{table}

\subsection{Using invariant}

An alternative way to build a list of orbit representatives is 
to use  invariants. Success for invariant based approach is not guaranteed
for two reasons~: small orbits are hidden and difficult to detect, and
the invariants used may not be discriminating enough ! Moreover, invariant
approach does not give orbit sizes and even less the generator
set of stabilizers. The invariant approach proposed in \cite{IACRNOTE} failed to find a list of
representatives of $\boole(3,4,7)$. In that case,  the number of orbits is $\cls(3,4,7)= 68433$ 
and using invariants, the authors got 68095 classes whence missing 338 orbits. 

\subsection{Counting near bent functions}

Let us recall that a 7-bit Boolean function is near bent when its Walsh spectrum 
takes three values 0, $\pm 16$. Such a function has degree
less or equal to 4. The set of near bent functions is invariant under the action of affine general linear group. From the classification of $\boole(3,4,7)$, 
it is possible to count the number of near bent functions. For each $f\in \boole(3,4,7)$, we determine the number $N(f)$ of quadratic forms $q\in \boole(2,2,7)$ such that $f+q$ is near bent.
By this naive approach, one find
the total number of near-bent functions in seven variables:

$$
	\sum_{f\in \boole(3,4,7)/ \level(2)}  N(f) \times \frac{\sharp \aglm}{\sharp\stableveldim(2,7,f) }=88624918554694407235840 \approx 2^{76.3}
$$

In Table \ref{STAB}, we can read the number of classes of $\boole(s,t,m)$ whose the stabilizer has a small order. For example, there are 50308 classes of $\boole(3,4,7)$ with a stabilizer of order 1 that represents 74\% of classes.
\begin{table}[h]
	\caption{\label{STAB}Multiplicities of small order stablizers.}
	\begin{tabular}{|l|r|r|r|r|r|r|r|r|r|r|}
\hline
	order &1    &2 &3 &4 &6 &7 &8 &12 &14 &16\\
\hline
\hline
   $\boole(3,4,7)$   &50308& 9591& 134& 3059& 235& 12& 1877& 163& 15& 895\\
  $\boole(4,7,7)$   &389  &571 &7 &444 &48 &3 &384 &68 &7&236\\
\hline
\end{tabular}
\end{table}
It is not reasonable 
to store the  full classfication of $\boole(2,4,7)$ 
simply because  the number of classes 
is huge ~: $\cls(2,4,7) = 118 140 881 980$.
However, we can adapt the descending method to classify the set of near bent functions $f+q$ where $\sharp\stableveldim(2,7,f)>1$.
Finally, we obtain  4243482 classes of near bent functions in $\boole(2,4,7)$. Note that
 99.2\% of classes have a trivial stabilizer.  The classification
of the near bent function of $\boole(2,4,7)$ is availble on the website 
of the project. We hope that all the data presented here can be used to 
answer the following
open problems :

\begin{open} 
	It is well known that the restriction to any hyperplane
	of a bent function is near bent. 
	Are all the near bent function a restriction of a bent function~?
\end{open}

\begin{open} 
	As suggested in note \cite{IACRNOTE},  
	is it feasible to count/classify the 8-bit bent
	function from the classification 7-bit near bent functions ?
\end{open}

\subsection{Covering radius of \rmc(3,7)}

In 2019,  Wang  \cite {WANG} proved  that the covering radius of
$\rmc(2,7)$ is equal to 40. A part of that proof, is based
on the classification of $\boole(2,6,6)$.  The covering radius
of $\rmc(3,7)$ into $\rmc(4,7)$ is known to be 20 see \cite{MR4395429}. 
In the recent preprint \cite{COVERING}, Gao, Kan, Li and Wang showed 
the covering radius of $\rmc(3,7)$ is less or equal to 20 using the 
classification of $\boole(4,6,6)$.  All these
methods use more or less computer assistance. Here, we point out how to use
directely the classification of $\boole(4,7,7)$ to obtain
that the covering radius of $\rmc(3,7)$ is less or equal to $20$.  The key point is to use a variation
of Leon's algorithm to exibit small weight codewords in
the translate of a code.

Given the generator matrix $G$ of an $[n,k]$- Reed-Muller code
$C$, the algorithm \texttt{distance(f, G, T)} 
applies a random procedure to check the existence of a Boolean function
of weight less or equal to $T$ in the translate code $f+C$. This algorithm uses three components :
\begin{itemize}
 \item \texttt{action(f)} returns a random action of $\aglm$ on \texttt{f}
 \item \texttt{pivoting(G)} applies Gauss elimination algorithm to the generator matrix  \texttt{G} choosing a random pivot on each of its line. Each line of the matrix obtained has weight less or equal to $n-k+1$
 \item \texttt{reduce(g,G)} transforms \texttt{g} adding to it the lines of \texttt{G} corresponding of the pivot position. More precisely, for each line $L_i$ of \texttt{G}, let us denote $p_i$ the position of the pivot on this line, \texttt{reduce} adds to \texttt{g} the line $L_i$ when \texttt{g}$(p_i)=1$. It appears that the weight of \texttt{g} after reduction is at most $n-k$.
\end{itemize}
The algorithm finishes when it finds a Boolean function \texttt{g} in the translate code of weight less or equal to \texttt{T} or when the number of trials exceeds an arbitrary limit \texttt{maxIter}.

We apply \texttt{distance(f, G, T)} to each representative of $\boole(4,4,7)$ to prove the
non-existence of Boolean functions  at 
distance greater than 20 from $\rmc(3,7)$. This work requires
an average of 538.6 trials with standard deviation
806.17.

\begin{center}
	\begin{lstlisting}[float,style=gcc, xleftmargin={20mm},caption={Counting trials to find a small cosetword.},linewidth={110mm},numbers=left, emph={pop,push,empty}]
maxIter = 2048
Algorithm  distance( f, G, T )
{
// G  generator matrix of a [n,k]-Reed-Muller code
	score  = n
	trials = 0
	while ( score > T ) and ( trails < maxIter) {
		g := action( f )
		pivoting( G )
		reduce( g, G )
		w := weight( g )
		if ( w < score )
			score := w
		trials = trials + 1
	}
	return trials;
}
\end{lstlisting}
\end{center}

\begin{open}
	In \cite{MR4395429}, the covering radius of $\rmc(4,8)$ in $\rmc(5,8)$ 
	is shown to be 26. Is it possible to build a classification of $\boole(5,6,8)$ and
	to apply similar methods in order to  determine the covering radius 
	of $\rmc(4,8)$ or at least in $\rmc(6,8)$ ?
\end{open}

\section{Conclusion}

We present an efficient descending method to classify the cosets of Reed-Muller codes. 
This procedure allow us to obtain the classification of two important
cosets of length 128~: $\rmq(4,7,2)$ and $\rmq(7,7,3)$. The first one
provides the classification of near  bent functions in seven variable. From the
the second, we explain how we refind the value of the covering radius of $\rmc(3,7)$.


\bibliographystyle{plain} 

\begin{thebibliography}{1}

\bibitem{MR4395429}
Randall Dougherty, R.~Daniel Mauldin, and Mark Tiefenbruck.
\newblock The covering radius of the {R}eed-{M}uller code {$RM(m-4,m)$} in
  {$RM(m-3,m)$}.
\newblock {\em IEEE Trans. Inform. Theory}, 68(1):560--571, 2022.

\bibitem{COVERING}
J.~Gao, H.~Kan, Y.~Li, and Q.~Wang.
\newblock The covering radius of the third-order reed-muller codes ${\rm
  rm}(3,7)$ is 20.
\newblock {\em submitted to IEEE IT}, 2023.

\bibitem{PROJECT}
Valérie Gillot and Philippe Langevin.
\newblock Classification of $b(s,t,7)$.
\newblock \url{http://langevin.univ-tln.fr/project/agl7/aglclass.html}, 2022.

\bibitem{AGLHOU}
Xiang-Dong Hou.
\newblock {${\rm AGL}(m,2)$} acting on {$R(r,m)/R(s,m)$}.
\newblock {\em J. Algebra}, 171(3):921--938, 1995.

\bibitem{BOOK}
Xiang-Dong Hou.
\newblock {\em Lectures on finite fields}, volume 190 of {\em Graduate Studies
  in Mathematics}.
\newblock American Mathematical Society, Providence, RI, 2018.

\bibitem{MAIORANA}
James~A. Maiorana.
\newblock A classification of the cosets of the {R}eed-{M}uller code
  {$R(1,6)$}.
\newblock {\em Math. Comp.}, 57(195):403--414, 1991.

\bibitem{IACRNOTE}
Meng Qingshu, Zhang Huanguo, Cui Jingsong, and Yang Min.
\newblock Almost enumeration of eight-variable bent functions.
\newblock {\em iacr preprint}, 2005.

\bibitem{SERESS}
\'{A}kos Seress.
\newblock {\em Permutation group algorithms}, volume 152 of {\em Cambridge
  Tracts in Mathematics}.
\newblock Cambridge University Press, Cambridge, 2003.

\bibitem{WANG}
Qichun Wang.
\newblock The covering radius of the {R}eed-{M}uller code {$RM(2,7)$} is 40.
\newblock {\em Discrete Math.}, 342(12):111625, 7, 2019.

\end{thebibliography}

\end{document}